\documentclass[journal,twoside,web]{ieeecolor}
\usepackage{lcsys}
\usepackage{algorithmic}
\usepackage{graphicx}
\usepackage{textcomp}
\def\BibTeX{{\rm B\kern-.05em{\sc i\kern-.025em b}\kern-.08em
    T\kern-.1667em\lower.7ex\hbox{E}\kern-.125emX}}
\usepackage{hyperref}
\hypersetup{hidelinks}
\usepackage{amsfonts,amssymb}
\usepackage{mathrsfs}
\usepackage[cmex10]{amsmath}
\usepackage{mathtools}
\usepackage{bbm}
\DeclareMathAlphabet{\mathbbb}{U}{bbold}{m}{n}
\usepackage{array}
\usepackage{physics}
\usepackage{cite}

\usepackage{amsthm}
\newtheorem{theorem}{Theorem}
\newtheorem{lemma}{Lemma}
\newtheorem{definition}{Definition}

\newtheorem{condition}{Condition}
\newtheorem{proposition}{Proposition}

\newtheorem{remark}{Remark}

\newcommand{\bbC}{\mathbb{C}}

\newcommand{\calC}{\mathcal{C}}
\newcommand{\calN}{\mathcal{N}}

\newcommand{\calS}{\mathcal{S}}
\newcommand{\calT}{\mathcal{T}}
\newcommand{\calA}{\mathcal{A}}

\newcommand{\hermconj}  {^{\mathsf{H}}}
\newcommand{\trans}     {^{\mathsf{T}}}
\newcommand{\pha}[1]    {\underline{#1}}
\newcommand{\phaconj}[1]{\overline{\underline{#1}}}
\newcommand{\vect}[1]   {\boldsymbol{#1}}
\newcommand{\phavec}[1] {\underline{\boldsymbol{#1}}}

\newcommand{\mat}[1]    {\boldsymbol{#1}}
\newcommand{\phamat}[1] {\underline{\boldsymbol{#1}}}
\newcommand{\re}[1]     {\mathscr{R}(#1)}
\newcommand{\im}[1]     {\mathscr{I}(#1)}

\begin{document}
\title{Nonlinear Stability of Complex Droop Control in Converter-Based Power Systems}
\author{Xiuqiang He, \IEEEmembership{Member, IEEE}, Verena Häberle, \IEEEmembership{Student Member, IEEE}, Irina Subotić, \IEEEmembership{Student Member, IEEE}, and Florian Dörfler, \IEEEmembership{Senior Member, IEEE}
\thanks{This work was supported by the European Union’s Horizon 2020
research and innovation program under Grant 883985.}
\thanks{The authors are with the Automatic Control Laboratory, ETH Zurich, Switzerland. Email:\{xiuqhe,verenhae,subotici,dorfler\}@ethz.ch.}
}

\maketitle
\thispagestyle{empty}
\pagestyle{empty}

\begin{abstract}
In this letter, we study the nonlinear stability problem of converter-based power systems, where the converter dynamics are governed by a complex droop control. This complex droop control augments the well-known power-frequency (p-f) droop control, and it proves to be equivalent to the state-of-the-art dispatchable virtual oscillator control (dVOC). In this regard, it is recognized as a promising grid-forming solution to address the high penetration of converters in future power systems. In previous work, the global stability of dVOC (i.e., complex droop control) has been proven by prespecifying a nominal synchronous steady state. For a general case of non-nominal (i.e., drooped) synchronous steady states, however, the stability problem requires further investigation. In this letter, we provide parametric conditions under which a non-nominal synchronous steady state exists and the system is almost globally asymptotically stable with respect to this non-nominal synchronous steady state.
\end{abstract}


\section{Introduction}

\IEEEPARstart{T}{HE} electric power system is currently undergoing a huge transformation caused by the replacement of conventional synchronous generation with converter-interfaced renewable energy sources. Future power systems may operate with 100\% converter-based generation \cite{Chen-100}, where converters must shoulder the responsibility of forming the grid voltage (phase and amplitude). A class of control strategies to handle this responsibility is called grid-forming control \cite{Rosso-GFM-review}. Grid-forming control methods are typically developed from the standpoint of \textit{a single converter device} \cite{Lu-param-tuning}. For \textit{a multi-converter interconnected system}, however, the stability of the system is a big concern, especially with respect to frequency synchronization and voltage stabilization.

Power-frequency (p-f) droop control represents a class of elementary grid-forming control, where the operation characteristics of synchronous generators are mimicked \cite{droop-control}. The stability of droop-controlled interconnected systems has been widely studied in the past \cite{Dorfler-kuramoto-transient-stability, Simpson-Porco, Schiffer-cell-structure, Schiffer-droop, Simpson-voltage}. In particular, since stability at the system level, especially global stability, is generally less considered during the control design stage \cite{droop-control}, it is nontrivial in general cases to guarantee global stability. Therefore, additional particular assumptions such as fixed voltage amplitudes \cite{Dorfler-kuramoto-transient-stability, Simpson-Porco, Schiffer-cell-structure}, a lossless network 
\cite{Simpson-Porco, Schiffer-cell-structure, Schiffer-droop}, or a simplified network representation \cite{Simpson-voltage} have been applied in the existing stability studies \cite{Dorfler-kuramoto-transient-stability, Simpson-Porco, Schiffer-cell-structure, Schiffer-droop, Simpson-voltage} for droop control.

Dispatchable virtual oscillator control (dVOC) is an advanced grid-forming control strategy, which has been developed recently by a top-down design \cite{Colombino-dVOC, Gross-dVOC, Subotic-dVOC} as well as a bottom-up design \cite{Lu-param-tuning}, where the global stability guarantee has been rigorously derived and experimentally validated. Moreover, dVOC has been increasingly recognized as one of the most promising grid-forming controls \cite{Lu-voc}. As presented in our recent work \cite{He-cplx-freq-sync}, dVOC can be considered as an augmentation of the standard p-f droop control. In particular, we have revealed the equivalence of dVOC to a \textit{complex droop control} from the perspective of complex frequency. The concept of complex frequency \cite{Milano-complex-freq} is an extension of classical frequency (or angular velocity) since the rates of change of both phase and amplitude are included. Aimed at the concepts of complex-frequency synchronization and classical voltage stabilization, we presented a \textit{linear stability analysis} in \cite{He-cplx-freq-sync}. This linear stability analysis is tractable and useful, but it is approximate since some assumptions are required to arrive at the linear problems of complex-frequency synchronization and voltage stabilization. A \textit{nonlinear stability analysis} for dVOC was undertaken in \cite{Colombino-dVOC, Gross-dVOC, Subotic-dVOC}, which, however, addressed only a prespecified nominal synchronous steady state and applied only to networks with a uniform $r/\ell$ ratio. Nonlinear stability analysis for the case of non-nominal (i.e., drooped) synchronous steady states and non-uniform networks, in turn, is not available in the literature so far.

In this letter, we address the nonlinear stability analysis of complex droop control for non-nominal synchronous steady states and non-uniform networks. We first leverage the linear result for complex-frequency synchronization \cite{He-cplx-freq-sync} to provide a condition for the existence of a non-nominal synchronous steady state. We then provide parametric conditions for this synchronous steady state being almost globally asymptotic stable. The stability results extend the previous ones in \cite{Gross-dVOC} to allow drooped operating points and general power networks with arbitrary impedance characteristics. The nonlinear stability results in this work are rigorous compared to our linear results in \cite{He-cplx-freq-sync}. In \cite{He-cplx-freq-sync}, the time-scale separation as well as the linear dc flow approximation were applied to make the problems linear. The nonlinear results in this work are independent of these assumptions. Moreover, the stability results are physically intuitive, providing engineering insights into power system operation and stability assessment.

The remainder of this section recalls some basic notation. The stability problem is formulated in Section \ref{sec-problem-formul} and then addressed in Section \ref{sec-stab-analy}. Case studies are shown in Section \ref{case-studies}. Section \ref{conclusion} concludes this letter.

\textit{Notation:} The set of complex numbers is denoted by $\bbC$. An underline indicates that a variable is complex, and $\re{\cdot}$ and $\im{\cdot}$ denote its real and imaginary parts, respectively. For a complex scalar $\pha x \in \bbC$, $\phaconj x$ denotes its complex conjugate. For a complex matrix $\phamat A \in \bbC^{m\times n}$, $\phamat A \trans$ and $\phamat A \hermconj$ denote its transpose and Hermitian transpose, respectively. For a real scalar $x$, a complex scalar $\pha x$, a real vector $\vect x$, and a complex vector $\phavec x$, $\abs{x}$, $\abs{\pha x}$, $\norm{\vect x}$, $\norm{\phavec x}$ denote the absolute value, the modulus, the Euclidean norm $(\vect x\trans \vect x)^{1/2}$, and the Euclidean norm $(\phavec x\hermconj \phavec x)^{1/2}$, respectively. For a matrix $\phamat A$, its induced 2-norm is denoted by $\norm{\phamat A}$. For a vector $\phavec x$, $\mathrm{diag} \left(\phavec x \right)$ denotes the diagonal matrix formed from it. The distance of a point $\phavec x $ to a set $\calC$ is denoted by $\norm{\phavec x}_\calC \coloneqq \min_{\phavec z \in \calC} \norm{\phavec z - \phavec x}$.

\section{Modeling and Stability Problem Statement}
\label{sec-problem-formul}

We consider a converter-based three-phase power system, where all nodes are interconnected by a resistive-inductive network. The converter nodes are modeled as grid-forming voltage sources \cite{Colombino-dVOC}. The load nodes are represented by constant impedance, which is a common assumption used for analytical stability studies \cite{Colombino-dVOC}\footnote{Since not all loads can be accurately represented by constant impedance, our results may not be accurate for other types of loads, particularly in the large-signal regime. It has not been well studied how to incorporate various loads (e.g., constant-power/-current loads and even dynamic loads) into nonlinear stability analysis, especially when considering voltage dynamics.}. The system is assumed to be three-phase balanced, and thus we can work in $\alpha\beta$ coordinates or the associated complex vector form.

\subsection{Power Network}

When ignoring the network dynamics, we obtain a static network representation. We further use the Kron reduction to eliminate the load nodes, obtaining a reduced network \cite{Dorfler-Kron-red}, where the set of converter nodes is denoted by $\calN = \{1, \cdots, N\}$, and the reduced network admittance matrix is denoted by $\phamat{Y} \in \bbC^{N \times N}$. We remark that the reduced network generally contains shunt branches at each node. Since the shunt loads can be absorbed into the power setpoints, we consider $\phamat{Y}$ as a complex-valued and symmetric Laplacian matrix without loss of generality \cite{He-cplx-freq-sync}. To each converter node $k \in \calN$, we associate a terminal voltage $\pha{v}_k \in \bbC$ and a converter output current $\pha{i}_{o,k} \in \bbC$. Based on the admittance matrix, the network equation can be formulated as
\begin{equation}
\label{static-network-eq}
    \phavec{i}_{o} = \phamat{Y}\, \phavec{v},
\end{equation}
where $\phavec{i}_{o} \coloneqq [\pha{i}_{o,1},\cdots,\pha{i}_{o,N}] \trans$ and $\phavec{v} \coloneqq [\pha{v}_{1},\cdots,\pha{v}_{N}] \trans$. Denote the complex voltage as $\pha{v}_k \coloneqq v_k e^{j\theta_{k}}$ with nonzero amplitude $v_k \neq 0$ and rotational phase angle $\theta_{k}$, and denote the output complex power as $\pha{s}_k \coloneqq {p_k} + j {q_k} \coloneqq \pha{v}_k \phaconj{i}_{o,k}$. We obtain the complex power-flow equations as
\begin{equation}
\label{power-flow-old}
    \pha{s}_k = \textstyle\sum\nolimits_{l = 1}^N \phaconj{y}_{kl} \pha{v}_k \phaconj{v}_l,
\end{equation}
where $\pha{y}_{kl}$ is the $k$th row and $l$th column entry of $\phamat{Y}$. Referring to \cite{He-cplx-freq-sync}, we introduce the definition of \textit{complex angle} and \textit{normalized complex power (conjugate)} respectively as
\begin{align}
\label{complex-angle-def}
    \pha {\vartheta}_k &\coloneqq \ln{v_k} + j\theta_k \ \Rightarrow \ \pha{v}_k = e^{\pha {\vartheta}_k},\\
\label{complex-power-def}
    \phaconj{\varsigma}_k &\coloneqq \phaconj{s}_k/v_k^2 = (p_k  - jq_k)/v_k^2 = {\pha{i}}_{o,k}/{\pha{v}_k}.
\end{align}
The power-flow equations in \eqref{power-flow-old} are then rewritten as
\begin{equation}
\label{power-flow-new}
    \phaconj{\varsigma}_k = \textstyle\sum\nolimits_{l = 1}^N {\pha{{y}}_{kl}} \frac{\pha{v}_l}{\pha{v}_k} = \textstyle\sum\nolimits_{l = 1}^N {{\pha{{y}}_{kl}}{e^{\pha{\vartheta}_{l} - \pha{\vartheta}_{k}}}}.
\end{equation}

\subsection{dVOC and Complex Droop Control}

\subsubsection{dVOC} When applying dVOC to the converters, their complex-valued terminal voltage behavior is given as \cite{Gross-dVOC}
\begin{equation}
\label{dvoc}
    {\pha{\dot v}_k} = j{\omega}_0 {\pha{v}_k} + \eta {e^{j\varphi }}\bigl( \phaconj{\varsigma}_k^{\star} {\pha{v}_k} - {\pha{i}}_{o,k} \bigr) + \eta \alpha \Phi_k(\pha{v}_k) {\pha{v}_k},
\end{equation}
where $\eta, \alpha > 0$ are control gains, $\phaconj{\varsigma}_k^{\star} \coloneqq (p_k^{\star}  - jq_k^{\star})/{{v_k^{\star 2}}}$ denotes the setpoint for the normalized power in \eqref{complex-power-def} with $p_k^{\star}$, $q_k^{\star}$, and $v_k^{\star}$ being the setpoints for active power, reactive power, and voltage amplitude, respectively, $e^{j\varphi}$ with $\varphi \in [0,\pi/2]$ is the rotation operator to adapt to different network impedance characteristics, and $\Phi_k(\pha{v}_k) \coloneqq ({v_k^{\star 2} - \lvert{\pha{v}_k}\rvert^2})/{v_k^{\star 2}}$ denotes the voltage regulation error. The development of dVOC was inspired by consensus synchronization \cite[Prop. 1]{Colombino-dVOC}. Briefly speaking, the first term in \eqref{dvoc} induces a harmonic oscillator at the nominal frequency, the second term synchronizes the relative phases to the power setpoints via current feedback, and the third term regulates the voltage amplitude.

\subsubsection{Complex Droop Control} We define the \textit{complex frequency} by taking the time derivative of $\pha{\vartheta}_k$ in \eqref{complex-angle-def} as
\begin{equation}
\label{complex-freq-def}
    \pha {\varpi}_k \coloneqq \dot{\pha{\vartheta}}_k = \dot v _k/v_k + j{\dot \theta _k} = {\dot{\pha{v}}_k}/{\pha{v}_k}.
\end{equation}
We then use $\dot{\pha{\vartheta}}_k = {\dot{\pha{v}}_k}/{\pha{v}_k}$ to transform \eqref{dvoc} into \textit{complex-angle coordinates}, obtaining a complex droop control as
\begin{equation}
    \label{complex-droop}
    \dot{\pha {\vartheta}}_k = \pha {\varpi}_0 + \eta e^{j\varphi} \left( \phaconj{\varsigma}_k^{\star} - \phaconj{\varsigma}_k\right) + \eta \alpha \tfrac{{v_k^{\star 2} - e^{2\re{\pha{\vartheta}_k}}}}{{v_k^{\star 2}}},
\end{equation}
where $\pha {\varpi}_0 \coloneqq j{\omega}_0$ denotes the nominal complex frequency, the second term denotes the droop gain multiplied by the imbalance of complex power, and the third term stabilizes the voltage amplitude. We remark that complex angle and complex frequency are emerging concepts recently developed in \cite{Milano-complex-freq}. These novel concepts enable a uniform representation of the phase and amplitude dynamics driven by the imbalance of complex power and thus facilitate stability analysis regarding frequency synchronization and voltage stabilization \cite{He-cplx-freq-sync}.

With $\varphi = \pi/2$, the complex droop control in \eqref{complex-droop} resembles a standard p-f and q-v droop control \cite{Colombino-dVOC}. Despite the droop-like behavior, \eqref{complex-droop} differs largely from the standard droop dynamics during large transients. For the standard droop control, the global stability guarantee has only been rigorously established for lossless networks and fixed voltage amplitudes \cite{Schiffer-cell-structure, Simpson-Porco}. In comparison, \eqref{complex-droop} is a multivariable control \cite{He-cplx-freq-sync} that allows going beyond such decoupled scenarios and obtaining a more general global stability guarantee.

\subsection{Statement of the Nonlinear Stability Problem}

The model of the dVOC-based nonlinear power system is obtained by interconnecting the network equation \eqref{static-network-eq} and the converter node dynamics in \eqref{dvoc} as
\begin{equation}
\label{nonlinear-syst}
    \dot {\phavec{v}} = \pha {\varpi}_0 \mat{I}_N {\phavec{v}} + \eta {e^{j\varphi }}( {\phamat{K} - \phamat{Y}} ) {\phavec{v}} + \eta \alpha \mat{\Phi}( {\phavec{v}}) {\phavec{v}},
\end{equation}
where $\phamat{K} \coloneqq {\rm diag}\bigl( \{ \phaconj{\varsigma}_k^{\star} \}_{k = 1}^N \bigr)$, $\mat{\Phi}({\phavec{v}}) \coloneqq {\rm diag} \bigl(\{\Phi_k(\pha{v}_k)\}_{k = 1}^N \bigr)$, and $\mat{I}_N$ is an identity matrix. In \eqref{nonlinear-syst}, we employ complex-valued differential equations to accommodate the complex-valued network admittance $\phamat{Y}$.

For the system \eqref{nonlinear-syst}, a nominal synchronous steady state is a state, where the frequencies synchronize to $\omega_0$, and the voltage amplitudes stabilize at their setpoints. To yield these power and voltage setpoints, they must satisfy the power-flow equations \cite{Gross-dVOC}. Otherwise, the system will operate in \textit{a non-nominal (drooped) synchronous steady state}, where \textit{the frequencies synchronize to a value deviating from $\omega_0$, and the voltage amplitudes stabilize at drooped points away from their setpoints}. In this work, we are concerned with the stability of the system \eqref{nonlinear-syst} with respect to a non-nominal synchronous steady state.

\section{Stability Analysis}
\label{sec-stab-analy}

We define an auxiliary system composed of the first two linear terms in the original system \eqref{nonlinear-syst} as
\begin{equation}
\label{linear-syst}
    \dot {\phavec{v}} = \phamat{A}\, {\phavec{v}}, \quad \phamat{A} \coloneqq \pha {\varpi}_0 \mat{I}_N + \eta {e^{j\varphi }}({\phamat{K} - \phamat{Y}}),
\end{equation}
which is equivalent to disabling the voltage regulation term in \eqref{nonlinear-syst} by letting $\alpha = 0$. In complex-angle coordinates, the dynamics in \eqref{linear-syst} corresponds to a complex droop control \eqref{complex-droop} without amplitude regulation, i.e., \eqref{complex-droop} with $\alpha = 0$,
\begin{equation}
    \label{complex-droop-without-u}
    \dot{\pha {\vartheta}}_k = \pha {\varpi}_0 + \eta e^{j\varphi} \left( \phaconj{\varsigma}_k^{\star} - \phaconj{\varsigma}_k\right).
\end{equation}
In the following, we first recall the prior result on complex-frequency synchronization of the linear system \eqref{linear-syst}. Based on this, we then show new results on the stability problem of the original nonlinear system \eqref{nonlinear-syst}.

\subsection{Complex-Frequency Synchronization}

\begin{definition} [Complex-frequency synchronization \cite{He-cplx-freq-sync}]
\label{complex-freq-sync-def}
The voltage trajectories $\phavec{v}$ in a connected network achieve complex-frequency synchronization if all the complex frequencies $\dot{\pha{\vartheta}}_k$ converge to a common constant complex frequency $\pha{\varpi}_{\rm{sync}}$, i.e., $\dot{\pha{\vartheta}}_k \to \pha{\varpi}_{\rm{sync}}, \, t \rightarrow \infty,\, \forall k \in \calN$.
\end{definition}

\begin{remark}
\label{remark-on-sync}
Complex-frequency synchronization implies both angular-frequency synchronization, $\dot \theta _k \to \im{\pha{\varpi}_{\rm{sync}}}$, and rate-of-change-of-voltage synchronization, $\dot{{v}}_k / {v}_k \to \re{\pha{\varpi}_{\rm{sync}}}$, $\forall k \in \calN$, where the voltage amplitudes are allowed to change but with the same exponential rate $\re{\pha{\varpi}_{\rm{sync}}}$. This novel concept of synchronization is relevant to investigate the phase-amplitude coupled dynamics in \eqref{linear-syst}.
\end{remark}

We denote the eigenvalues of the state matrix $\phamat{A}$ as $\pha{\lambda}_1$, $\pha{\lambda}_2$, $\cdots$, $\pha{\lambda}_N$, where $\re{\pha{\lambda}_1} \ge \re{\pha{\lambda}_2} \ge \cdots \ge \re{\pha{\lambda}_N}$. We term ${\pha{\lambda}_1}$ the dominant eigenvalue and assume that it has algebraic multiplicity one and $\re{\pha{\lambda_1}} > \re{\pha{\lambda_2}}$ \cite{He-cplx-freq-sync}. This is a generic assumption, reflecting the fact that power systems should have only one fundamental frequency component. If $\pha{\lambda}_1$ has algebraic multiplicity $m > 1$, then there will be multiple fundamental frequency components as $e^{\pha{\lambda}_1 t},\, te^{\pha{\lambda}_1 t},\, \cdots,\, t^{m-1}e^{\pha{\lambda}_1 t}$ in the response of the linear system \eqref{linear-syst}. We denote the eigenvector of $\pha{\lambda_1}$ as ${\phavec{\phi }_1}$ (the same for left and right eigenvectors since $\phamat{A}$ is symmetric) and denote the eigenspace spanned by ${\phavec{\phi }_1}$ as a set $\calS$, i.e.,
\begin{equation*}
    \calS \coloneqq \bigl\{\phavec{v} \in \bbC^N\, \big\vert \, \phavec{v} = \pha{\mu}\, \phavec{\phi }_1,\, \pha{\mu} \in \bbC \bigr\}.
\end{equation*}

We now recall a parametric stability condition for the auxiliary system \eqref{linear-syst}, which guarantees that the system achieves complex-frequency synchronization on the eigenspace $\calS$.

\begin{condition}[\hspace{1sp}\protect{\cite[Condition 2]{He-cplx-freq-sync}}]
\label{condition-sync}
There exists a maximal phase difference $\bar{\delta} \in [0, \pi/2)$ and a maximal voltage-amplitude ratio deviation $\bar{\gamma} \in (0, 1)$ such that $\lvert{\theta_k - \theta_l}\rvert \le \bar{\delta}$ and $\lvert {v_k}/{v_l} - 1 \rvert \le \bar{\gamma}$ hold in any synchronous state for all $k,l \in \calN$. Moreover, the power setpoints $\phaconj{\varsigma}_k^{\star}$, the rotation operator ${e^{j\varphi}}$, and the network admittance matrix $\phamat{Y}$ satisfy
\begin{equation}
\label{condition-sync-ineq}
    \max\limits_k \re{{e^{j\varphi}}\phaconj{\varsigma}_k^{\star}} < \tfrac{{1 + \cos \bar {\delta}}}{2}{\left(1 - \bar{\gamma} \right)^2} {\lambda _2} \bigl(\re{{e^{j\varphi}} \phamat{Y}}\bigr),
\end{equation}
where $\re{{e^{j\varphi}} \phamat{Y}}$ is a Laplacian matrix, formed by the real part of the entries of ${e^{j\varphi}} \phamat{Y}$, and ${\lambda _2} \bigl(\re{{e^{j\varphi}} \phamat{Y}}\bigr)$ denotes its second largest eigenvalue (a positive real number).
\end{condition}

\begin{theorem}[\hspace{1sp}\protect{\cite[Thm. 2]{He-cplx-freq-sync}}]
\label{theorem-sync}
    Under Condition \ref{condition-sync}, for almost all initial states ${\phavec{v}_0}$ except for the non-generic initial condition ${\phavec{\phi}_1\trans}{\phavec{v}_0} = 0$, the system \eqref{linear-syst} achieves complex-frequency synchronization at the donimant eigenvalue, i.e., $\pha{\varpi}_{\rm{sync}} = {\pha{\lambda}_1}$, and the voltages $\phavec{v}$ converge to the eigenspace $\calS$.
\end{theorem}

The interpretation of Condition \ref{condition-sync} is deferred to later (after the closely related Condition \ref{condition-stability}). As shown in \cite{He-cplx-freq-sync}, Condition \ref{condition-sync} gives that $\re{\pha{\lambda}_k} < 0,\, \forall k \geq 2$. Thus, Theorem \ref{theorem-sync} follows intuitively from the fact that the linear dynamics in \eqref{linear-syst} are dominated by the dominant mode $e^{{\pha{\lambda}_1}t}$ while the other non-dominant modes decay to zero. When viewed from complex-angle coordinates, the complex-frequency synchronization is an outcome of the complex droop control in \eqref{complex-droop-without-u}.

Let $\pha{\gamma}_{lk} \coloneqq \pha{\gamma}_{l1}/\pha{\gamma}_{k1}$ denotes the complex-voltage ratios for all $l,k \in \calN$, where $\pha{\gamma}_{k1} \coloneqq \pha{v}_k/\pha{v}_1$, $\forall k \in \calN$, is the ratio to the node $k = 1$. The power-flow equations in \eqref{power-flow-new} can then be rewritten as $\phaconj{\varsigma}_k = \textstyle\sum\nolimits_{l = 1}^N {\pha{{y}}_{kl}} \pha{\gamma}_{lk}$. Since $\phavec{v}$ converges to $\calS$, $\pha{\gamma}_{lk}$ converges to a constant $\pha{\phi}_{l1}/\pha{\phi}_{k1}$, where $\pha{\phi}_{k1}$ denotes the $k$th entry of $\phavec{\phi}_{1}$ and $\pha{\phi}_{k1} \neq 0$ is guaranteed by Condition \ref{condition-sync} since ${v_k}/{v_l}$ is nonzero and bounded. It follows that both $\pha{\gamma}_{lk}$ and $\phaconj{\varsigma}_k$ remain unchanged during complex-frequency synchronization. For a complex-frequency synchronous state, we emphasize that it is not a steady state unless the real part of the complex frequency converges to zero.

\subsection{Existence of Non-Nominal Synchronous Steady State}

Consider the original system \eqref{nonlinear-syst} with amplitude regulation enabled. We impose some constraints on the voltage setpoints to guarantee the existence of a synchronous steady state. These constraints are imposed according to the complex-frequency synchronous state without amplitude regulation.

\begin{condition}
\label{condi-setpoints}
Let Condition \ref{condition-sync} hold. Consider the complex-frequency synchronization of the complex droop control $\dot{\pha {\vartheta}}_k = \pha {\varpi}_0 + \eta e^{j\varphi} \bigl( \phaconj{\varsigma}_k^{\star} - \textstyle\sum\nolimits_{l = 1}^N {\pha{{y}}_{kl}} \pha{\gamma}_{lk} \bigr)$, where $\dot{\pha {\vartheta}}_k$ synchronizes to $\pha{\varpi}_{\rm{sync}} = {\pha{\lambda}_1}$ and $\pha{\gamma}_{lk}$ converges to $\pha{\phi}_{l1}/\pha{\phi}_{k1}$. When enabling voltage amplitude regulation, assume that the ratios between the voltage setpoints are consistent with the voltage amplitude ratios under the complex-frequency synchronization, i.e., $v_l^{\star}/v_k^{\star} = \lvert \pha{\phi}_{l1} \rvert / \lvert \pha{\phi}_{k1} \rvert, \, \forall k,l \in \calN$.
\end{condition}

In Condition \ref{condi-setpoints}, $\pha{\varpi}_{\rm sync}$ can be drooped. The particular case, $\pha{\varpi}_{\rm sync} = \pha {\varpi}_0$, leads to a nominal synchronous steady state, implicitly expressed by the power-flow equations $\phaconj{\varsigma}_k^{\star} = \textstyle\sum\nolimits_{l = 1}^N \pha{y}_{kl}\pha{\gamma}_{lk}$ with $ \lvert {\pha{\gamma}_{lk}} \rvert = v_l^{\star}/v_k^{\star}$, equivalent to \cite[Condition 1]{Gross-dVOC} prescribed for a nominal synchronous steady state. Therefore, Condition \ref{condi-setpoints} relaxes \cite[Condition 1]{Gross-dVOC}. However, the voltage setpoint constraints in Condition \ref{condi-setpoints} are still restrictive, and not as arbitrary as power setpoints. Therefore, not all conceivable synchronous steady states are encompassed in Condition \ref{condi-setpoints}. When it comes to the benefit from Condition \ref{condi-setpoints}, it not only guarantees the existence of a synchronous steady state but also directly provides the concrete steady-state values for frequency and voltage (as well as power outputs, without the need for power-flow calculations).

Consider the system \eqref{nonlinear-syst} with $\alpha \neq 0$ and with consistent voltage setpoints as in Condition \ref{condi-setpoints}. We define steady-state voltage amplitudes by a set $\calA$ as 
\begin{equation*}
    \calA \coloneqq \bigl\{\phavec{v} \in \bbC^N\, \big \vert \, \lvert \pha{v}_k \rvert = v_k^{\star}\sqrt{1 + {\re{\pha{\lambda}_1}}/(\eta\alpha)},\, \forall k\in \calN \bigr\},
\end{equation*}
where we assume that $1 + {\re{\pha{\lambda}_1}}/(\eta\alpha) > 0$ to avoid an ill-posed problem. We then define a set of synchronous steady states by a compact set $\calT$ as
\begin{equation*}
    \calT \coloneqq \calS \cap \calA.
\end{equation*}

\begin{proposition}
\label{prop-existence}
    Under Condition \ref{condi-setpoints}, the system \eqref{nonlinear-syst} has a synchronous steady state in the set $\calT$, where the synchronous frequency is $\dot{\theta}_k = \im{\pha{\lambda}_1}$, and the voltage amplitude steady state is $v_k = v_k^{\star}\sqrt{1 + \re{\pha{\lambda}_1} /(\eta\alpha)},\, k \in \calN$.
\end{proposition}

\begin{proof}
We first claim that $\calT$ is nonempty under Condition \ref{condi-setpoints}. This is true because the voltage amplitudes in $\calA$ fit the eigenvector in $\calS$, i.e., $\lvert \pha{v}_l \rvert/\lvert \pha{v}_k \rvert = v_l^{\star}/v_k^{\star} = \lvert \pha{\phi}_{l1} \rvert / \lvert \pha{\phi}_{k1} \rvert, \forall k,l \in \calN$. We then show that the set $\calT$ satisfies frequency synchronization and voltage steady state. Substituting $\forall \phavec{v} \in \calT$ into the system dynamics \eqref{nonlinear-syst}, we obtain that 
\begin{equation}
    \begin{aligned}
        \dot {\phavec{v}} &= \phamat{A}\, {\phavec{v}} + \eta \alpha {\rm diag} \bigl(\bigl\{ \tfrac{{v_k^{\star 2} - \lvert{\pha{v}_k}\rvert^2}}{v_k^{\star 2}} \bigr\}_{k = 1}^N \bigr) {\phavec{v}}\\
        &= \phamat{A}\, {\phavec{v}} - \eta \alpha \tfrac{\re{\pha{\lambda}_1}}{\eta \alpha } {\phavec{v}}\\
        &= \pha{\lambda}_1 {\phavec{v}} - \re{\pha{\lambda}_1} {\phavec{v}} 
        = \im{\pha{\lambda}_1} {\phavec{v}}, \quad \forall \phavec{v} \in \calT,
    \end{aligned}
\end{equation}
where the second equality holds with the voltage amplitudes specified in $\calA$, and the third equality holds due to $\phavec{v} \in \calS$. The last equality indicates that both the frequencies $\dot{\theta}_k = \im{\pha{\lambda}_1}$ and the voltage amplitudes remain invariant in $\calT$.
\end{proof}

\subsection{Conditions for Stability}

Next, we provide a parametric stability condition for the system \eqref{nonlinear-syst} and then show an analytical stability result.

\begin{condition}
\label{condition-stability}
Let Conditions \ref{condition-sync} and \ref{condi-setpoints} hold. In terms of Condition \ref{condition-sync}, it is further assumed that
\begin{equation}
\label{condition-stability-ineq}
    \max\limits_k \re{{e^{j\varphi }}\phaconj{\varsigma}_k^{\star}} + \alpha < \tfrac{{1 + \cos \bar {\delta}}}{2}{\left(1 - \bar{\gamma} \right)^2} {\lambda _2} \bigl(\re{{e^{j\varphi }} \phamat{Y}}\bigr).
\end{equation}
\end{condition}

Compared to \eqref{condition-sync-ineq} in Condition \ref{condition-sync}, \eqref{condition-stability-ineq} in Condition \ref{condition-stability} is strengthened with the appearance of the voltage regulation gain $\alpha$ on the left-hand side. This coincides with the fact that the system \eqref{nonlinear-syst} is augmented with voltage regulation compared to \eqref{linear-syst}. In Condition \ref{condition-sync} or \ref{condition-stability}, the assumptions on the phase differences and the voltage ratios in the synchronous (steady) state are reasonable and reflect operational constraints since power systems are supposed to operate in a healthy state with close node voltages and phase angles. Moreover, ${\lambda _2} \bigl(\re{{e^{j\varphi }} \phamat{Y}}\bigr)$ denotes the algebraic connectivity of the graph corresponding to the real part of the rotated admittance matrix. The inequality \eqref{condition-stability-ineq} quantifies the margin of stability, for which the network should be sufficiently well connected and not be heavily loaded, and the voltage regulation should not be too fast. We relate \eqref{condition-stability-ineq} to the previous conditions in \cite{Colombino-dVOC, Gross-dVOC, Subotic-dVOC} in the sense that \eqref{condition-stability-ineq} extends the previous ones to non-synchronous steady states as well as non-uniform networks. In particular, the term $\re{{e^{j\varphi }} \phamat{Y}}$ quantifies how the real part of the rotated admittance matrix matters for stability, which unveils a new stability factor for networks with non-uniform $r/\ell$ ratios.

We define almost global asymptotic stability with respect to a set \cite{Colombino-dVOC} before presenting the stability result. For power systems, this notion implies that almost all initial states lead to a synchronous steady state, except for a negligible (zero-measure) of initial states leading to voltage collapse.

\begin{definition}[Almost globally asymptotic stability]
A dynamic system is almost globally asymptotically stable with respect to a set $\calT$ if it is Lyapunov stable with respect to $\calT$ and for all initial states, except those contained in a zero measure set, the trajectories converge to $\calT$.
\end{definition}

\begin{theorem}
\label{theorem-stability}
    Under Condition \ref{condition-stability}, the system \eqref{nonlinear-syst} is almost globally asymptotically stable with respect to the set $\calT$.
\end{theorem}

\begin{proof}
We use a similar proof as in \cite[Sec. IV-C]{Gross-dVOC}. The main difference from the proof in \cite{Gross-dVOC} is that we need to adapt the Lyapunov function to non-nominal synchronous steady states. Moreover, we work with complex variables to handle the complex-valued admittance matrix $\phamat{Y}$.

A matrix $\phamat{P} \coloneqq \mat{I}_N - {\phavec{\phi }_1} {\phavec{\phi }_1\hermconj}/({\phavec{\phi }_1\hermconj} {\phavec{\phi }_1})$ is defined as the projector onto the subspace orthogonal to $\calS$. The distance of $\phavec{v}$ to the set $\calS$ is then given by $\norm{\phavec{v}} _{\calS} = \norm{\phamat{P}\,\phavec{v}}$. From $\phamat{P}\hermconj = \phamat{P}$ and $\phamat{P}^2 = \phamat{P}$, it follows that $\phavec{v}\hermconj \phamat{P}\, \phavec{v} = \norm{\phamat{P}\,\phavec{v}}^2 = \norm{\phavec{v}} _{\calS} ^2$. We consider a Lyapunov function candidate $V$ as
\begin{equation*}
    V \coloneqq \frac{1}{2} \phavec{v}\hermconj \phamat{P}\, \phavec{v} + \frac{1}{2} \eta\alpha \alpha_1 \textstyle\sum\limits_{k=1}^{N}{\Bigl( \frac{\re{\pha{\lambda}_1}}{\eta \alpha} v_k^{\star} + \frac{v_k^{\star 2} - \lvert{\pha{v}_k} \rvert^2}{v_k^{\star}} \Bigr)^2}.
\end{equation*}
where the constant $\alpha_1$ is given by
\begin{equation}
\label{alpha1-def}
    \alpha_1 \coloneqq {\eta c}/\bigl({5 \norm{\phamat{A} - \pha{\lambda}_1 \mat{I}_N}^2}\bigr) > 0
\end{equation}
with the constant $c > 0$ defined by
\begin{equation}
\label{c-def}
    c \coloneqq \tfrac{{1 + \cos \bar {\delta}}}{2}{\left(1 - \bar{\gamma} \right)^2} {\lambda _2} \bigl(\re{{e^{j\varphi }} \phamat{Y}}\bigr) - \max\limits_k \re{{e^{j\varphi }}\phaconj{\varsigma}_k^{\star}} - \alpha.
\end{equation}
We observe that $V(\phavec{v})$ is positive definite and radially unbounded with respect to the compact set $\calT = \calS \cap \calA$. We derive the derivative of $V(\phavec{v})$ along the dynamics in \eqref{nonlinear-syst} as
\begin{align*}
    \dot{V} &= \tfrac{1}{2} \phavec{v}\hermconj (\phamat{A}\hermconj \phamat{P} + \phamat{P}\, \phamat{A}) \phavec{v} + \tfrac{1}{2} \eta \alpha \phavec{v}\hermconj (\mat{\Phi} \phamat{P} + \phamat{P} \mat{\Phi}) \phavec{v} \\
    & \quad\, -2\alpha_1 \textstyle\sum\nolimits_{k=1}^{N}{ \bigl( \re{\pha{\lambda}_1} + \eta\alpha \Phi_k \bigr) \re{ \dot{\pha{v}}_k \phaconj{v}_k} },
\end{align*}
where $\frac{{\rm d}}{{\rm d}t}{ (\lvert {\pha{v}_k} \rvert^2) } = \frac{{\rm d}}{{\rm d}t}{ (\pha{v}_k \phaconj{v}_k) } = \dot{\pha{v}}_k \phaconj{v}_k + \pha{v}_k \dot{\phaconj{v}}_k = 2\re{ \dot{\pha{v}}_k \phaconj{v}_k }$. We further obtain from $\phamat{A} {\phavec{\phi }_1} = {\pha{\lambda}_1} {\phavec{\phi }_1}$, $\phamat{P} = \mat{I}_N - {\phavec{\phi }_1} {\phavec{\phi }_1\hermconj}/({\phavec{\phi }_1\hermconj} {\phavec{\phi }_1})$ that $\phamat{A} - {\pha{\lambda}_1} \mat{I}_N = (\phamat{A} - {\pha{\lambda}_1}\mat{I}_N)\phamat{P}$, and then that $\phamat{A} = \phamat{A}\,\phamat{P} + {\pha{\lambda}_1} (\mat{I}_N - \phamat{P})$. It follows that $\phamat{P}\, \phamat{A} = \phamat{P}\, \phamat{A}\, \phamat{P}$ and $\phamat{A} \hermconj \phamat{P} = \phamat{P}\, \phamat{A}\hermconj \phamat{P}$. Using these dependencies, we can rewrite $\dot{V}$, and then bound it using Lemmas \ref{lemma-aux1} and \ref{lemma-aux2} in Appendix as follows
\begin{align}
    \dot{V} &= \tfrac{1}{2} \phavec{v}\hermconj \phamat{P} (\phamat{A}\hermconj + \phamat{A}) \phamat{P}\, \phavec{v} + \tfrac{1}{2} \eta \alpha \phavec{v}\hermconj (\mat{\Phi} \phamat{P} + \phamat{P} \mat{\Phi}) \phavec{v} \notag \\
    & \quad\, -2\alpha_1 \textstyle\sum\nolimits_{k=1}^{N}{ \bigl( \re{\pha{\lambda}_1} + \eta\alpha \Phi_k \bigr) \re{ \dot{\pha{v}}_k \phaconj{v}_k} } \notag \\
    &\leq \tfrac{1}{2} \phavec{v}\hermconj \phamat{P} (\phamat{A}\hermconj + \phamat{A} + 2 \eta \alpha) \phamat{P}\, \phavec{v} \notag \\
    & \quad\, -2\alpha_1 \textstyle\sum\nolimits_{k=1}^{N}{ \bigl( \re{\pha{\lambda}_1} + \eta\alpha \Phi_k \bigr) \re{ \dot{\pha{v}}_k \phaconj{v}_k} } \notag \\
    \label{v-dot-expression}
    &\leq - \eta c \norm{\phavec{v}} _{\calS} ^2 - 2\alpha_1 \textstyle\sum\nolimits_{k=1}^{N}{ \bigl( \re{\pha{\lambda}_1} + \eta\alpha \Phi_k \bigr) \re{ \dot{\pha{v}}_k \phaconj{v}_k} }.
\end{align}
We further bound the summation term in \eqref{v-dot-expression} as
\begin{align}
    &\textstyle\sum\nolimits_{k=1}^{N}{ \bigl( \re{\pha{\lambda}_1} + \eta\alpha \Phi_k \bigr) \re{ \dot{\pha{v}}_k \phaconj{v}_k} } \notag \\
    &= \mathscr{R} \Bigl[ \phavec{v}\hermconj \bigl( \re{\pha{\lambda}_1} \mat{I}_N + \eta\alpha \mat{\Phi} \bigr) \dot {\phavec{v}} \Bigr] \notag \\
    &= \mathscr{R} \Bigl[ \phavec{v}\hermconj \bigl( \re{\pha{\lambda}_1} \mat{I}_N + \eta\alpha \mat{\Phi} \bigr) \bigl( \phamat{A} + \eta\alpha \mat{\Phi} \bigr) \phavec{v} \Bigr] \notag \\
    &= \mathscr{R} \Bigl[\phavec{v}\hermconj \bigl( \re{\pha{\lambda}_1} \mat{I}_N + \eta\alpha \mat{\Phi} \bigr)  (\phamat{A} - {\pha{\lambda}_1}\mat{I}_N)  \phamat{P}\, \phavec{v}\Bigr] \notag \\
    & \quad\, + \mathscr{R} \Bigl[\phavec{v}\hermconj \bigl( \re{\pha{\lambda}_1} \mat{I}_N + \eta\alpha \mat{\Phi} \bigr) \bigl( {\pha{\lambda}_1}\mat{I}_N + \eta\alpha \mat{\Phi} \bigr) \phavec{v}\Bigr] \notag \\
    &= \mathscr{R} \Bigl[\phavec{v}\hermconj \bigl( \re{\pha{\lambda}_1} \mat{I}_N + \eta\alpha \mat{\Phi} \bigr)  (\phamat{A} - {\pha{\lambda}_1}\mat{I}_N)  \phamat{P}\, \phavec{v}\Bigr] \notag \\
    & \quad\, + \phavec{v}\hermconj \bigl( \re{\pha{\lambda}_1} \mat{I}_N + \eta\alpha \mat{\Phi} \bigr) \bigl( \re{\pha{\lambda}_1}\mat{I}_N + \eta\alpha \mat{\Phi} \bigr) \phavec{v} \notag \\
    &\geq -\bigl \lvert\phavec{v}\hermconj \bigl( \re{\pha{\lambda}_1} \mat{I}_N + \eta\alpha \mat{\Phi} \bigr)  (\phamat{A} - {\pha{\lambda}_1}\mat{I}_N)  \phamat{P}\, \phavec{v}\bigr \rvert \notag \\
    & \quad\, + \norm{\bigl( \re{\pha{\lambda}_1} \mat{I}_N + \eta\alpha \mat{\Phi} \bigr)\phavec{v}}^2 \notag \\
    &\geq -\norm{\bigl( \re{\pha{\lambda}_1} \mat{I}_N + \eta\alpha \mat{\Phi} \bigr)\phavec{v}} \norm{\phamat{A} - {\pha{\lambda}_1}\mat{I}_N} \norm{\phavec{v}}_{\calS} \notag \\
    \label{sum-term}
    & \quad\, + \norm{\bigl( \re{\pha{\lambda}_1} \mat{I}_N + \eta\alpha \mat{\Phi} \bigr)\phavec{v}}^2.
\end{align}
We next show that $\forall \phavec{v} \in \bbC^N$, it holds that
\begin{equation}
\label{V-negative-def}
    \dot{V} \leq -\alpha_1 \left(\norm{\phamat{A} - {\pha{\lambda}_1}\mat{I}_N} \norm{\phavec{v}}_{\calS} + \norm{\bigl( \re{\pha{\lambda}_1} \mat{I}_N + \eta\alpha \mat{\Phi} \bigr)\phavec{v}} \right)^2.
\end{equation}
We substitute \eqref{sum-term} into \eqref{v-dot-expression} and then notice that the following inequality suffices to show \eqref{V-negative-def},
\begin{equation*}
\begin{aligned}
    &- \eta c \norm{\phavec{v}} _{\calS} ^2 + 2\alpha_1\norm{\bigl( \re{\pha{\lambda}_1} \mat{I}_N + \eta\alpha \mat{\Phi} \bigr)\phavec{v}} \norm{\phamat{A} - {\pha{\lambda}_1}\mat{I}_N} \norm{\phavec{v}}_{\calS}\\
    &- 2\alpha_1\norm{\bigl( \re{\pha{\lambda}_1} \mat{I}_N + \eta\alpha \mat{\Phi} \bigr)\phavec{v}}^2\\
    &\leq -\alpha_1 \left(\norm{\phamat{A} - {\pha{\lambda}_1}\mat{I}_N} \norm{\phavec{v}}_{\calS} + \norm{\bigl( \re{\pha{\lambda}_1} \mat{I}_N + \eta\alpha \mat{\Phi} \bigr)\phavec{v}} \right)^2.
\end{aligned} 
\end{equation*}
This inequality is equivalent to $\vect{x}(\phavec{v}) \trans \mat{Q} \vect{x}(\phavec{v})  \geq 0,\, \forall \phavec{v} \in \bbC^N$, where $\vect{x}(\phavec{v}) \coloneqq \bigl[ \norm{\phavec{v}}_{\calS},\, \norm{\bigl( \re{\pha{\lambda}_1} \mat{I}_N + \eta\alpha \mat{\Phi} \bigr)\phavec{v}} \bigr]\trans$ and 
\begin{equation*}
    \mat{Q} \coloneqq \begin{bmatrix}
        \eta c - \alpha_1 \norm{\phamat{A} - {\pha{\lambda}_1}\mat{I}_N}^2 & -2\alpha_1 \norm{\phamat{A} - {\pha{\lambda}_1}\mat{I}_N}\\
        -2\alpha_1 \norm{\phamat{A} - {\pha{\lambda}_1}\mat{I}_N} & \alpha_1
    \end{bmatrix}.
\end{equation*}
With $\alpha_1 > 0$ given in \eqref{alpha1-def}, $\mat{Q}$ is positive semidefinite since the Schur complement gives $\eta c - 5 \alpha_1 \norm{\phamat{A} - {\pha{\lambda}_1}\mat{I}_N}^2 \geq 0$. We now conclude that \eqref{V-negative-def} holds. Further, we have $\dot{V} = 0$ if and only if $\norm{\phavec{v}}_{\calS} = 0$ and $\norm{\bigl( \re{\pha{\lambda}_1} \mat{I}_N + \eta\alpha \mat{\Phi} \bigr)\phavec{v}} = 0$. Both the set $\calT$ and the origin $\mathbbb{0}_N$ satisfy $\dot{V} = 0$. We refer to \cite{Gross-dVOC} for the result that $\mathbbb{0}_N$ is an unstable equilibrium and its region of attraction is a zero measure set. Namely, the initial states in this region of attraction will lead the voltage to collapse to zero ($\mathbbb{0}_N$). Hence, the set $\calT$ is almost globally asymptotically stable for almost all initial states except those contained in this zero measure set \cite[Thm. 1]{Gross-dVOC}.
\end{proof}

Following up on Theorem \ref{theorem-stability}, we can further consider the network dynamics and the voltage and current control dynamics. In this case, we can extend the stability analysis for non-nominal synchronous steady states to multiple time scales by applying nested singular perturbations \cite{Subotic-dVOC}.

\section{Case Studies}
\label{case-studies}

\begin{figure}
  \begin{center}
  \includegraphics[width=7cm]{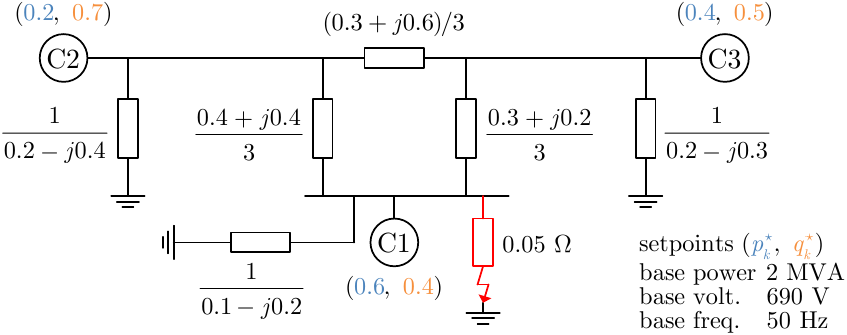}
  \caption{{A three-bus test system with inconsistent power setpoints and non-uniform $r/\ell$ ratios.}}
  \label{test-system}
  \end{center}$\vspace{-3mm}$
\end{figure}

\begin{figure}
  \begin{center}
  \includegraphics[width=7.3cm]{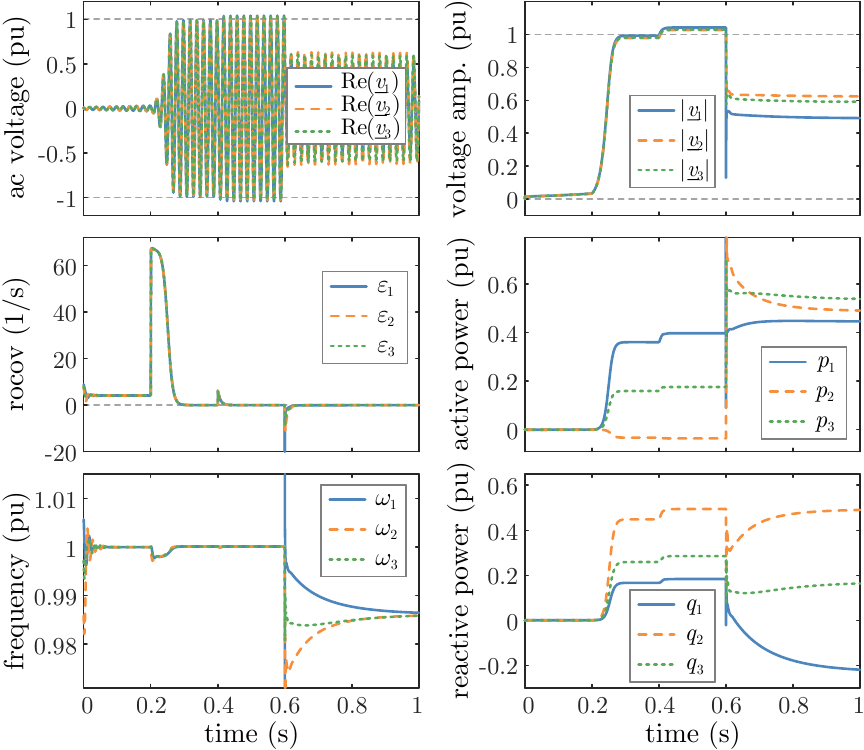}
  \caption{Simulation of a black start, followed by voltage regulation enabled at $0.2\ \rm{s}$, voltage setpoints scaling up to $105\%$ uniformly at $0.4\ \rm{s}$, and a short-circuit fault occurring at $0.6\ \rm{s}$.}
  \label{simulation}
  \end{center}$\vspace{-7mm}$
\end{figure}

We illustrate the theoretical results by an electromagnetic transient (EMT) simulation on a three-bus system in Fig.~\ref{test-system}, where the system model, parameters, and control gains remain identical to those in \cite{He-cplx-freq-sync}. We calculate the synchronous complex frequency $\pha{\varpi}_{\rm sync}$ by $\pha{\lambda}_1$ as $4.13 + j314.19\, {\rm rad/s}$, and then specify $(v_1^{\star},v_2^{\star},v_3^{\star}) = (0.9612,0.9469,0.9498)$ using Condition \ref{condi-setpoints}. Consider the constraint on the steady state as $\bar{\delta} = \pi/6 \in [0, \pi/2)$ as well as $\bar{\gamma} = 0.2 \in (0, 1)$. By identifying that Condition \ref{condition-stability} holds, we conclude that the system is almost globally asymptotically stable.

The system black-starts from a point close to the origin. In Fig.~\ref{simulation}, we observe that complex-frequency synchronization is achieved during $0$ to $0.2\ \rm{s}$. In this period, voltage regulation is disabled, and the voltage amplitudes increase exponentially. At $0.2\ \rm{s}$, we enable the voltage regulation, after which the voltage amplitudes are lifted rapidly close to their setpoints (but with deviations due to the steady state being drooped). At $0.4\ \rm{s}$, we uniformly scale the voltage setpoints up to $105\%$ (Condition \ref{condi-setpoints} still holds), and then the voltages converge to a new steady state while the frequency settles down to the original steady state. At $0.6\ \rm{s}$, a short-circuit fault occurs. A virtual impedance-based current-limiting strategy is employed to prevent the converters from overcurrent. The fault branch and the introduced virtual impedance alter the equivalent network, and then the original voltage setpoints are no longer consistent with Condition \ref{condi-setpoints}. We note, however, that the system remains stable in a new synchronous steady state due to the robustness of the droop-like behavior.

We provide more insights into stability under fault conditions. When a grid fault occurs, the grid enters into an abnormal operation stage, during which converters should maintain grid-forming operation with auxiliary control/protection strategies. When a converter is disconnected due to overcurrent/overvoltage, or when a fault line trips, the system changes accordingly, but its model can remain in the same form as under normal conditions. System operators need to screen out critical contingencies and assess the stability of the corresponding systems. The linear results in \cite{He-cplx-freq-sync} and the nonlinear ones of this work can be utilized for this purpose.

\section{Conclusion}
\label{conclusion}

We investigate the nonlinear stability problem of grid-forming complex droop control (i.e., dVOC) in converter-based power systems. We provide parametric conditions for almost globally asymptotic stability of complex droop control with respect to non-nominal synchronous steady states, which apply to networks with non-uniform $r/\ell$ ratios. The conditions quantify the operational requirements for a global stability guarantee, thus providing practical guidelines for the stable operation of converter-based power systems (typically microgrids). The results also suggest that complex droop control has better stability properties than the standard droop control. Our future work will address the relaxation of Condition \ref{condi-setpoints} to all conceivable synchronous steady states.

\section*{Appendix}

\begin{lemma}
\label{lemma-aux1}
    Under Condition \ref{condi-setpoints}, $\forall \phavec{v} \in \bbC^N$, it holds that
    \begin{equation}
    \label{lemma-aux1-ineq}
    \phavec{v}\hermconj \bigl(\mat{\Phi} \phamat{P} + \phamat{P} \mat{\Phi} \bigr) \phavec{v}
    \leq 2 \phavec{v}\hermconj \phamat{P}\, \phavec{v}.
    \end{equation}
\end{lemma}

\begin{proof}
We rewrite \eqref{lemma-aux1-ineq} in real-valued variables as
\begin{equation}
\label{lemma-aux1-eq-ineq}
    \vect{v}\trans \bigl( \mat{\Psi} \mat{P} + \mat{P} \mat{\Psi} \bigr) \vect{v} = 2\vect{v}\trans \mat{P} \mat{\Psi} \vect{v} \leq 2\vect{v}\trans \mat{P} \vect{v}
\end{equation}
where $\vect{v} \coloneqq \re{\phavec{v}} \otimes [1 \ 0]\trans + \im{\phavec{v}} \otimes [0\ 1]\trans$, $\mat{\Psi} \coloneqq \mat{\Phi} \otimes \mat{I}_2$, and $\mat{P} \coloneqq 
\mat{I}_{2N} - {\vect{\phi}_1} {\vect{\phi}_1 \trans} / \textstyle\sum\nolimits_{k=1}^N  \lvert \pha{\phi}_{k1} \rvert ^2$ with $\vect{\phi}_1 \coloneqq \re{\phavec{\phi}_1} \otimes \left[\begin{smallmatrix} 1 & 0\\ 0 & 1 \end{smallmatrix}\right] + \im{\phavec{\phi}_1} \otimes \left[\begin{smallmatrix} 0 & -1\\1 & 0 \end{smallmatrix}\right]$.
We expand $\vect{v}\trans \mat{P} \mat{\Psi} \vect{v}$ as
\begin{equation*}
    \vect{v}\trans \mat{P} \mat{\Phi} \vect{v} = \vect{v}\trans \mat{P} \vect{v} - \vect{v}\trans {\rm diag} \bigl(\bigl\{\tfrac{v_k^2}{v_k^{\star 2}} \mat{I}_2 \bigr\}_{k = 1}^N \bigr) \mat{P} \vect{v}.
\end{equation*} 
Since $v_l^{\star}/v_k^{\star} = \lvert \pha{\phi}_{l1} \rvert / \lvert \pha{\phi}_{k1} \rvert,\, \forall k,l \in \calN$, holds under Condition \ref{condi-setpoints}, there exists a positive real number $\zeta > 0$ such that $v_k^{\star} = \zeta \lvert \pha{\phi}_{k1} \rvert, \, \forall k \in \calN$. Then, the inequality in \eqref{lemma-aux1-eq-ineq} is equivalent to
\begin{align*}
    - \vect{v}\trans {\rm diag} \bigl(\bigl\{\tfrac{v_k^2}{v_k^{\star 2}} \mat{I}_2 \bigr\}_{k = 1}^N \bigr) \mat{P} \vect{v} &\leq 0,\\
    \vect{v}\trans \mat{P} \vect{v} - \zeta^{ 2}\vect{v}\trans {\rm diag} \bigl(\bigl\{\tfrac{v_k^2}{v_k^{\star 2}} \mat{I}_2 \bigr\}_{k = 1}^N \bigr) \mat{P} \vect{v} &\leq \vect{v}\trans \mat{P} \vect{v},\\
    \label{lemma-aux1-ineq1}
    \vect{v}\trans \mat{P} \vect{v} - \vect{v}\trans {\rm diag} \bigl(\bigl\{\tfrac{v_k^2}{\lvert \pha{\phi}_{k1} \rvert ^2} \mat{I}_2 \bigr\}_{k = 1}^N \bigr) \mat{P} \vect{v} &\leq \vect{v}\trans \mat{P} \vect{v}.
\end{align*}
The last inequality above is equivalent to \cite[Lemma 1]{Gross-dVOC}.
\end{proof}

\begin{lemma}
\label{lemma-aux2}
    Under Condition \ref{condition-stability}, $\forall \phavec{v} \in \bbC^N$, it holds that
    \begin{equation}
    \label{lemma-aux2-ineq}
        \tfrac{1}{2}\phavec{v}\hermconj \phamat{P} \bigl(\phamat{A}\hermconj + \phamat{A} + 2\eta\alpha \bigr) \phamat{P}\, \phavec{v} \leq - \eta c \norm{\phavec{v}} _{\calS} ^2.
    \end{equation}
\end{lemma}

\begin{proof}
With $c$ given in \eqref{c-def} under Condition \ref{condition-stability}, the inequality \eqref{lemma-aux2-ineq} is equivalent to
\begin{equation*}
\begin{aligned}
    &\tfrac{1}{2}\phavec{v}\hermconj \phamat{P} \bigl(\phamat{A}\hermconj + \phamat{A}) \phamat{P}\, \phavec{v} \leq  \\
    &\norm{\phavec{v}} _{\calS} ^2 \Bigl [\max\limits_k \re{{e^{j\varphi }}\phaconj{\varsigma}_k^{\star}} - \tfrac{{1 + \cos \bar {\delta}}}{2}{\left(1 - \bar{\gamma} \right)^2} {\lambda _2} \bigl(\re{{e^{j\varphi }} \phamat{Y}}\bigr) \Bigr].
\end{aligned}
\end{equation*}
This holds by referring to the result on complex-frequency synchronization in \cite[Thm. 2]{He-cplx-freq-sync}.
\end{proof}

\bibliographystyle{IEEEtran}
\bibliography{IEEEabrv,Bibliography}

\begin{thebibliography}{10}
\providecommand{\url}[1]{#1}
\csname url@samestyle\endcsname
\providecommand{\newblock}{\relax}
\providecommand{\bibinfo}[2]{#2}
\providecommand{\BIBentrySTDinterwordspacing}{\spaceskip=0pt\relax}
\providecommand{\BIBentryALTinterwordstretchfactor}{4}
\providecommand{\BIBentryALTinterwordspacing}{\spaceskip=\fontdimen2\font plus
\BIBentryALTinterwordstretchfactor\fontdimen3\font minus
  \fontdimen4\font\relax}
\providecommand{\BIBforeignlanguage}[2]{{%
\expandafter\ifx\csname l@#1\endcsname\relax
\typeout{** WARNING: IEEEtran.bst: No hyphenation pattern has been}%
\typeout{** loaded for the language `#1'. Using the pattern for}%
\typeout{** the default language instead.}%
\else
\language=\csname l@#1\endcsname
\fi
#2}}
\providecommand{\BIBdecl}{\relax}
\BIBdecl

\bibitem{Chen-100}
J.~Chen \emph{et~al.}, ``100\% converter-interfaced generation using virtual
  synchronous generator control: A case study based on the {Irish} system,''
  \emph{Electr. Power Syst. Res.}, vol. 187, p. 106475, 2020.

\bibitem{Rosso-GFM-review}
R.~Rosso, X.~Wang, M.~Liserre, X.~Lu, and S.~Engelken, ``Grid-forming
  converters: Control approaches, grid-synchronization, and future trends—a
  review,'' \emph{{IEEE} Open J. Ind. Appl.}, vol.~2, pp. 93--109, 2021.

\bibitem{Lu-param-tuning}
M.~Lu, S.~Dutta, V.~Purba, S.~Dhople, and B.~Johnson, ``A grid-compatible
  virtual oscillator controller: Analysis and design,'' in \emph{Proc. IEEE
  Energy Convers. Congr. Expo.}, 2019, pp. 2643--2649.

\bibitem{droop-control}
M.~C. Chandorkar, D.~M. Divan, and R.~Adapa, ``Control of parallel connected
  inverters in standalone ac supply systems,'' \emph{{IEEE} Trans. Ind. Appl.},
  vol.~29, no.~1, pp. 136--143, 1993.

\bibitem{Dorfler-kuramoto-transient-stability}
F.~D\"{o}rfler and F.~Bullo, ``Synchronization and transient stability in power
  networks and nonuniform {Kuramoto} oscillators,'' \emph{{SIAM} J. Control
  Optim.}, vol.~50, no.~3, pp. 1616--1642, 2012.

\bibitem{Simpson-Porco}
J.~W. Simpson-Porco, F.~Dörfler, and F.~Bullo, ``Synchronization and power
  sharing for droop-controlled inverters in islanded microgrids,''
  \emph{Automatica}, vol.~49, no.~9, pp. 2603--2611, 2013.

\bibitem{Schiffer-cell-structure}
J.~Schiffer, D.~Efimov, and R.~Ortega, ``Global synchronization analysis of
  droop-controlled microgrids—a multivariable cell structure approach,''
  \emph{Automatica}, vol. 109, p. 108550, 2019.

\bibitem{Schiffer-droop}
J.~Schiffer, R.~Ortega, A.~Astolfi, J.~Raisch, and T.~Sezi, ``Conditions for
  stability of droop-controlled inverter-based microgrids,'' \emph{Automatica},
  vol.~50, no.~10, pp. 2457--2469, 2014.

\bibitem{Simpson-voltage}
J.~W. Simpson-Porco, F.~D{\"o}rfler, and F.~Bullo, ``Voltage stabilization in
  microgrids via quadratic droop control,'' \emph{{IEEE} Trans. Autom.
  Control}, vol.~62, no.~3, pp. 1239--1253, 2017.

\bibitem{Colombino-dVOC}
M.~Colombino, D.~Groß, J.-S. Brouillon, and F.~Dörfler, ``Global phase and
  magnitude synchronization of coupled oscillators with application to the
  control of grid-forming power inverters,'' \emph{{IEEE} Trans. Autom.
  Control}, vol.~64, no.~11, pp. 4496--4511, 2019.

\bibitem{Gross-dVOC}
D.~Groß, M.~Colombino, J.-S. Brouillon, and F.~Dörfler, ``The effect of
  transmission-line dynamics on grid-forming dispatchable virtual oscillator
  control,'' \emph{{IEEE} Trans. Control Netw. Syst.}, vol.~6, no.~3, pp.
  1148--1160, 2019.

\bibitem{Subotic-dVOC}
I.~Subotić, D.~Groß, M.~Colombino, and F.~Dörfler, ``A {Lyapunov} framework
  for nested dynamical systems on multiple time scales with application to
  converter-based power systems,'' \emph{{IEEE} Trans. Autom. Control},
  vol.~66, no.~12, pp. 5909--5924, 2021.

\bibitem{Lu-voc}
M.~Lu, ``Virtual oscillator grid-forming inverters: State of the art, modeling,
  and stability,'' \emph{{IEEE} Trans. Power Electron.}, vol.~37, no.~10, pp.
  11\,579--11\,591, 2022.

\bibitem{He-cplx-freq-sync}
\BIBentryALTinterwordspacing
X.~He, V.~Häberle, and F.~Dörfler, ``Complex-frequency synchronization of
  converter-based power systems,'' 2022, submitted to \textit{IEEE Trans.
  Control Netw. Syst.} [Online]. Available:
  \url{https://arxiv.org/abs/2208.13860}
\BIBentrySTDinterwordspacing

\bibitem{Milano-complex-freq}
F.~Milano, ``Complex frequency,'' \emph{{IEEE} Trans. Power Syst.}, vol.~37,
  no.~2, pp. 1230--1240, 2022.

\bibitem{Dorfler-Kron-red}
F.~Dörfler and F.~Bullo, ``Kron reduction of graphs with applications to
  electrical networks,'' \emph{{IEEE} Trans. Circuits Syst. I-Regul. Pap.},
  vol.~60, no.~1, pp. 150--163, 2013.

\end{thebibliography}

\end{document}